\documentclass[10pt]{article}

\usepackage[margin=1in]{geometry}
\usepackage{lastpage}
\usepackage{graphicx,color}
\usepackage{amsmath,amsthm,amsfonts}
\usepackage[boxed]{algorithm}
\usepackage[noend]{algpseudocode}
\usepackage{hyperref}
\usepackage{xspace}
\usepackage{authblk}

\newcommand{\func}[1]{\textproc{#1}}

\newtheorem{claim}{Claim}

\newcommand{\T}{\ensuremath{T^*}}
\newcommand{\Z}{\mathcal{Z}}

\newcommand{\eat}[1]{}

\newcommand{\pos}{{\sc PoS}\xspace}
\newcommand{\poa}{{\sc PoA}\xspace}
\newcommand{\nash}{{\sc NE}\xspace}

\newcommand{\upp}[1]{\lceil c(#1) \rceil}
\newcommand{\low}[1]{\lfloor c(#1) \rfloor}

\newtheorem{fact}{Fact}
\newtheorem{theorem}{Theorem}
\newtheorem{lemma}[theorem]{Lemma}

\graphicspath{{./figs/}}

\begin{document}

\title{On the Price of Stability of Undirected Multicast Games}
\author{Rupert Freeman\thanks{Supported in part by NSF IIS-1527434 and ARO W911NF-12-1-0550.} }
\author{Samuel Haney\thanks{Supported in part by NSF Awards CCF-1527084 and CCF-1535972l.}}
\author{Debmalya Panigrahi\protect\footnotemark[2]}
\affil{Department of Computer Science, Duke University, Durham, NC 27708, USA \\ \texttt{\{rupert,shaney,debmalya\}@cs.duke.edu}}
\date{}

\maketitle

\begin{abstract}
  In multicast network design games, a set of agents choose paths from their source locations to a 
  common sink with the goal of minimizing their individual costs, where the cost of an edge is divided equally among the agents 
  using it. Since the work of Anshelevich et al. (FOCS 2004) that introduced network design games, the main open problem in this field 
  has been the price of stability (PoS) of multicast games. For the special case of broadcast games 
  (every vertex is a terminal, i.e., has an agent), a series of works has culminated in a 
  constant upper bound on the PoS (Bil\`{o} et al., FOCS 2013). However, no significantly sub-logarithmic 
  bound is known for multicast games. In this paper, we make progress toward resolving this question 
  by showing a constant upper bound on the PoS of multicast games for quasi-bipartite graphs. These are 
  graphs where all edges are between two terminals (as in broadcast games) or between a terminal and a 
  nonterminal, but there is no edge between nonterminals. This represents a natural class of intermediate 
  generality between broadcast and multicast games.
  In addition to the result itself, our techniques overcome some of the fundamental difficulties of analyzing the PoS of general multicast games, and are a promising step toward resolving this major open problem.
\end{abstract}

\section{Introduction}
\label{sec:introduction}

In cost sharing network design games, we are given a graph/network $G = (V,E)$ with edge costs and 
a set of users (agents/players) who want to send traffic from their respective source vertices to sink vertices.
Every agent must choose a path along which to route traffic, and the cost of every edge is shared 
equally among all agents having the edge in their chosen path, i.e., using the edge to route traffic.
This creates a \emph{congestion game} since the players benefit from other players choosing the same resources.
A Nash equilibrium is attained in this game when no agent has incentive to unilaterally deviate from her current routing path.
The social cost of such a game is the sum of costs of edges being used in at least one 
routing path, and efficiency of the game is measured by the ratio of the social cost in 
an equilibrium state to that in an optimal state. (The optimal state is defined as one where the 
social cost is minimized, but the agents need not be in equilibrium.) The maximum value of this 
ratio (i.e., for the most expensive equilibrium state) is called the {\em price of anarchy}
of the game, while the minimum value (i.e., for the least expensive equilibrium state) is 
called its {\em price of stability}. It is well known that even for the most restricted settings, 
the price of anarchy can be $\Omega(n)$ for $n$ agents 
(see Figure~\ref{fig:anarchy-lower-bound} for a simple example). Therefore, the main 
question of research interest has been to bound the price of stability (\pos) of this 
class of congestion games.

\begin{figure}[ht]
  \centering
  \includegraphics[scale=.6]{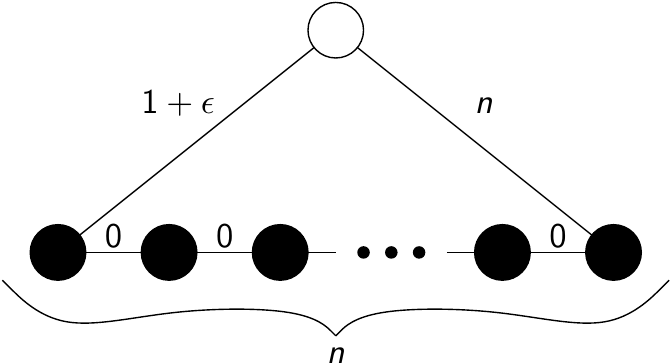}
  \caption{An example with a price of anarchy of $\Omega(n)$. Each black vertex is an agent, and the white vertex is the root (i.e. the common sink). There is a NE where every agent routes through the edge of weight $n$. Each agent has a cost of $1$ in such a configuration. On the other hand, the optimal configuration has a total cost of $1+\epsilon$ where every agent routes through the edge of cost $(1+\epsilon)$.}
  \label{fig:anarchy-lower-bound}
\end{figure}

Anshelevich~{\em et al.}~\cite{AnshelevichDKTWR08} introduced
network design games and obtained a bound of $O(\log n)$ on the \pos
in directed networks with arbitrary source-sink pairs. While this 
is tight for directed networks, they left determining tighter bounds on
the \pos in undirected networks as an open question. Subsequent work has focused
on the case of all agents sharing a common sink (called
{\em multicast games}) and its restricted subclass where every vertex
has an agent residing at it (called {\em broadcast games}). 
These problems are natural analogs of the Steiner tree and minimum
spanning tree (MST) problems in a game-theoretic setting. For broadcast games, 
Fiat~{\em et al.}~\cite{FiatKLOS06} improved the \pos bound to 
$O(\log\log n)$, which was subsequently improved to $O(\log\log\log n)$
by Lee and Ligett~\cite{LeeL13}, and ultimately to $O(1)$ by
Bil\`{o}, Flamminni, and Moscardelli~\cite{BiloFM13}. For multicast games,
however, progress has been much slower, and the only improvement over
the $O(\log n)$ result of Anshelevich~{\em et al.} is a bound of
$O(\log n / \log \log n)$ due to Li~\cite{Li09}. In contrast, the best 
known lower bounds on the PoS of both broadcast and multicast games
are small constants~\cite{BiloCFM13}.
Determining the \pos of multicast games has become one of the 
most compelling open questions in the area of network games. 

In this paper, we achieve progress toward answering this question. 
In the multicast setting, a vertex is said to be a {\em terminal} if it 
has an agent on it, else it is called a {\em nonterminal}. Note that
in the broadcast problem, there are no nonterminals and all the edges
are between terminal vertices. In this paper, we consider multicast
games in {\em quasi-bipartite} graphs: all edges are either
between two terminals, or between a nonterminal and a terminal.
(That is, there is no edge with both nonterminal endpoints.)
This represents a natural setting of intermediate generality between 
broadcast and multicast games. Moreover,
quasi-bipartite graphs have been widely studied for the Steiner tree
problem (see, e.g.,~\cite{RajagopalanV99,RobinsZ05,ChakrabartyDV11,ByrkaGRS13}) 
and has provided insights for the problem on general graphs. 
Our main result is an $O(1)$ bound on the \pos of multicast games in
quasi-bipartite graphs.
\begin{theorem}
\label{thm:main}
	The price of stability of multicast games in quasi-bipartite graphs
	is a constant.
\end{theorem}	
In Table~\ref{tab:prev}, we summarize the known and new results for single 
sink network design and the corresponding cost sharing games.
\begin{table}[!ht]
  \caption{Single sink Network Design: Optimization and Cost sharing games.\label{tab:prev}}
	\begin{tabular}{|c|c|c|c|}
		\hline
    \multicolumn{2}{|c|}{Optimization (Approximation factor)}  & \multicolumn{2}{c|}{Cost sharing game (Price of Stability)}  \\
		\hline
		MST & Poly-time solvable & Broadcast & $O(1)$~\cite{BiloFM13}\\
		\hline
		Quasi-bipartite Steiner & 1.22~\cite{ByrkaGRS13} & Quasi-bipartite Multicast & $O(1)$ {\bf [This paper]} \\		
		\hline
		Steiner Tree & 1.39~\cite{ByrkaGRS13} & Multicast & $O\left(\frac{\log n}{\log \log n}\right)$~\cite{Li09}\\
		\hline
  \end{tabular}
\end{table}

In addition to the result itself, our techniques overcome some of the fundamental difficulties of analyzing the PoS of general multicast games, and therefore represent a promising step toward resolving this important open problem.
To illustrate this point, we outline the salient features of our analysis below.

The previous \pos bounds for multicast games~\cite{AnshelevichDKTWR08,Li09}
are based on analyzing a potential function $\phi_e$ defined on each edge $e$ as 
its cost scaled by the harmonic of the number of agents using the edge, i.e., $\phi_e = cost(e) \cdot (1+1/2+1/3+\cdots+1/j)$ where $j$ is the number of terminals using $e$.
The overall potential is $\phi = \sum_e \phi_e$.
When an agent changes her routing path (called a {\em move}), 
this potential exactly tracks the change in her shared cost. If the move
is an {\em improving} one, then the shared cost of the agent decreases
and so too does the potential. As a consequence, for an arbitrary sequence 
of improving moves starting with the optimal Steiner tree, the potential 
decreases in each move until a Nash Equilibrium (\nash) is reached.
This immediately yields a \pos bound of $H(n) = O(\log n)$~\cite{AnshelevichDKTWR08}.
To see this, note that the potential of any configuration is bounded below by its cost, and above by its cost times $H(n)$.
Then, letting $S_{NE}$ be the Nash equilibrium state reached, and $T^*$ be the optimal routing tree, we have
\begin{equation*}
  c(S_{NE}) \le \phi(S_{NE}) \le \phi(T^*) \le H(n)c(T^*).
\end{equation*}
This bound was later improved to $O(\log n/\log\log n)$ by Li~\cite{Li09} with a similar but more careful accounting argument.

The previous \pos bounds for broadcast games~\cite{FiatKLOS06,LeeL13,BiloFM13}
use a different strategy. As in the case of multicast games, these results analyze 
a game dynamics that starts with an optimal solution (MST) and ends in an 
\nash. However, the sequence of moves is carefully 
constructed --- the moves are not arbitrary improving moves. 
At a high level, the sequence follows the same pattern in all the previous results for broadcast games:
\begin{enumerate}
  \item Perform a \emph{critical} move: Allow some terminal $v$ to switch its path to introduce a single new edge into the solution, that is not in the optimal routing tree and is adjacent to $v$. This edge is associated with $v$ and denoted $e_v$.
    Any edge introduced by the algorithm in any move other than a critical move uses only edges in the current routing tree, and edges in the optimal routing tree.
    Therefore, we only need to account for edges added by critical moves.
  \item Perform a sequence of moves to ensure that the routing tree is \emph{homogenous}. That is, the difference in costs of a pair of terminals is bounded by a function of the length of the path between them on the optimal routing tree. For example, suppose two terminals $w$ and $w'$ differ in cost by more than the length of the path between them in the optimal routing tree. Then the terminal with larger cost has an improving move that uses this path, and then the other terminal's path to the root. Such a move introduces only edges in the optimal routing tree.
  \item \emph{Absorb} a set of terminals around $v$ in the shortest path metric defined on the optimal tree: terminals $w$ replace their current strategy with the path in the optimal routing tree to $v$, and then $v$'s path to the root. If $w$ had an associated edge $e_w$, introduced via a previous critical move, it is removed from the solution in this step.
\end{enumerate}

The absorbing step allows us to account for the cost of edges added via critical moves, by arguing that vertices associated with critical edges of similar length must be well-separated on the optimal routing tree.
If edges $e_u$ and $e_v$ are not far apart, the second edge to be added would be removed from the solution via the absorbing step.

Homogeneity facilitates absorption:
Suppose $v$ has performed a critical move adding edge $e_v$, and let $w$ be some other terminal.
While $v$ pays $c(e_v)$ to use edge $e_v$, $w$ would only pay $c(e_v)/2$ to use $e_v$, since it would split the cost with $v$. 
That is, if $w$ bought a path to $v$ and then used $v$'s path to the root, it would save at least $c(e_v)/2$ over $v$'s current cost.
If the current costs paid by $v$ and $w$ are not too different, and the distance between $v$ and $w$ not too large, then such a move is improving for $w$.

The previous results differ in how well they can homogenize: the tighter the bound on the difference in costs of a pair of terminals as a function of the length of the path between them in the optimal routing tree, the larger the radius in the absorb step.
In turn, a larger radius of absorption establishes a larger separation between edges with similar cost, which yields a smaller (tighter) bound on the PoS.

This homogenization-absorption framework has not previously been extended to multicast games.
The main difficulty is that there can be nonterminals that are in the routing tree at equilibrium but are not in the optimal tree.
No edge incident on these vertices is in the optimal tree metric, and therefore these vertices cannot be included in the homogenization process.
So, any critical edge incident on such a vertex cannot be charged via absorption.
This creates the following basic problem: what metric can we use for the homogenization-absorption framework that will satisfy the following two properties?
\begin{enumerate}
  \item The metric is feasible -- the sum of all edge costs in (a spanning tree of) the metric is bounded by the cost of the optimal routing tree. These edges can therefore be added or removed at will, without need to perform another set of moves to pay for them (in contrast to critical edges). This allows us to homogenize using these edges.
  \item The metric either includes all vertices (as is the case with the optimal tree metric for broadcast games), or if there are vertices not included in the metric, critical edges adjacent to these vertices can be accounted for separately, outside the homogenization-absorption framework.
\end{enumerate}
We create such a metric for quasi-bipartite graphs, allowing us to extend the homogenization-absorption framework to multicast games.
Our metric is based on a dynamic tree containing all the terminals and a dynamic set of nonterminals.
We show that under certain conditions, we can include the shortest edge incident on a nonterminal vertex, even if it is not in the optimal routing tree, in this dynamic tree.
These edges are added and removed throughout the course of the algorithm.
Our new metric is now defined by shortest path distances on this dynamic tree: the optimal routing tree extended with these special edges.
We ensure homogeneity not on the optimal routing tree, but on this dynamic metric.
Likewise, absorption happens on this new metric.
We define the metric in such a way that the following hold:
\begin{enumerate}
  \item The metric is feasible. That is, the total cost of all edges in the dynamic tree is within a constant factor of the cost of the optimal tree.
  \item Consider some critical edge $e_v$ such that the corresponding vertex $v$ is not in the metric. That is, it was not possible to add the shortest edge adjacent to $v$ to the dynamic tree while keeping it feasible. Therefore, $v$ is at infinite distance from every other vertex in this metric, ruling out homogenization. Then, $e_v$ can be accounted for separately, outside the homogenization-absorption framework.
\end{enumerate}
For the remaining edges $e_v$ such that $v$ is in the metric, we account for them by using the homogenization-absorption framework.
Our main technical contribution is in creating this feasible dynamic metric, going beyond the use of static optimal metrics in broadcast games.
While the proof of feasibility currently relies on the quasi-bipartiteness of the underlying graph, we believe that this new idea of a feasible dynamic metric is a promising ingredient for multicast games in general graphs.

\subsection{Related Work}
\label{sec:related}

Recall that the upper bounds for the \pos are a (large) constant and $O(\frac{\log n}{\log \log n})$ for broadcast and multicast games, respectively.
The corresponding best known lower bounds
are 1.818 and 1.862 respectively by Bil\`{o} {\em et al.}~\cite{BiloCFM13}, 
leaving a significant gap, even for broadcast games. Moreover,
Lee and Ligett~\cite{LeeL13} show that obtaining superconstant lower bounds,
even for multicast games where they might exist, is beyond current techniques.
While this lends credence to the belief that the \pos of multicast games is 
$O(1)$, Kawase and Makino~\cite{KawaseM13} have shown that the potential 
function approach of Anshelevich~{\em et al.}~\cite{AnshelevichDKTWR08} 
cannot yield a constant bound on the \pos, even for broadcast games. In fact,
Bil\`{o} {\em et al.}~\cite{BiloFM13} used a different approach for broadcast
games, as do we for multicast games on quasi-bipartite graphs.

Various special cases of network design games have also been considered.
For small instances ($n=2, 3, 4$), both upper~\cite{ChristodoulouCLPS09} 
and lower~\cite{BiloB11} bounds have been studied.
\cite{ChristodoulouCLPS09} show upper bounds of 1.65 and $4/3$ for two and three players respectively.
For weighted players, 
Anshelevich {\em et al.}~\cite{AnshelevichDKTWR08} showed that pure Nash 
equilibria exist for $n=2$, but the possibility of a corresponding result 
for $n\geq 3$ was refuted by Chen and Roughgarden~\cite{ChenR09}, who also 
provided a logarithmic upper bound on the \pos. An almost matching lower 
bound was later given by Albers~\cite{Albers09}.
Recently, Fanelli {\em et al.}~\cite{FanelliLMS15}, showed that the \pos of network design games on undirected rings is 3/2.

Network design games have also been studied for specific dynamics. In 
particular, starting with an empty graph, suppose agents arrive online and 
choose their best response paths. After all arrivals, agents 
make improving moves until an \nash is reached. The worst-case inefficiency
of this process was determined to be poly-logarithmic by
Charikar~{\em et al.} \cite{CharikarKMNS08}, who also posed the question
of bounding the inefficiency if the arrivals and moves are arbitrarily 
interleaved. This question remains open. Upper and lower bounds for the 
strong \poa of undirected network design games have also been 
investigated~\cite{Albers09,EpsteinFM09}.
They show that the price of anarchy in this setting is $\Theta(\log n)$.

\section{Preliminaries}
\label{sec:prelims}

Let $G = (V,E)$ be an undirected edge-weighted graph and let $c(e)$ denote the cost of edge $e$. Let $U \subseteq V$ be a set of \emph{terminals} and $r \in U$. In an instance of a \emph{network design game}, each terminal $u$ is associated with a \emph{player}, or \emph{agent}, that must select a path from $u$ to $r$. We consider instances in which $G$ is \emph{quasi-bipartite}, that is no edge $e$ has two nonterminal end points.

A \emph{solution}, or \emph{state}, is a set of paths connecting each player to the root. Let $\mathcal{S}$ be the set of all possible solutions. For a solution $S$, a terminal $u$, and some subset $E'$ of the edges in the graph, let $c_u^{E'}(S) = \sum_{e \in E'} c(e)/n_e(S)$ be the cost paid by $u$ for using edges in $E'$, where $n_e(S)$ is the number of players using edge $e$ in state $S$. Let $p_u(S)$ be the set of edges used by $u$ to connect to the root in $S$ and let $c_u(S) = c_u^{p_u(S)}(S)$ be the total cost paid by $u$ to use those edges. For a nonterminal $v$, if every terminal $u$ with $v \in p_u(S)$ uses the same path from $v$ to the root then define $p_v(S)$ to be this path from $v$ to $r$, and $c_v(S) = c_u^{p_v(S)}(S)$.
Additionally, we will sometimes refer to the cost a vertex $v$ pays, even if $v$ is a nonterminal.
By this we mean $c_v(S)$.
For any vertex $v \in S$, let $e_v$ be the edge in $p_v(S)$ with $v$ as an endpoint.

Let $\Phi : \mathcal{S} \to \mathbb{R}_+$ be the potential function introduced by Rosenthal~\cite{Rosenthal73}, defined by
\begin{equation*}
  \Phi(S) = \sum_{e \in E} c(e) H_{n_e(S)} = c(e) \left(1+\frac{1}{2}+\cdots+\frac{1}{n_e(S)}\right).
\end{equation*}
Let $u \in U$ and suppose $S$ and $S'$ are states for which $p_v(S) = p_v(S')$ for all players $v \not= u$. Then $\Phi(S')-\Phi(S) = c_u(S')-c_u(S)$. In particular, if a single player changes their path to a path of lower cost, the potential decreases.

The goal of each player is to find a path of minimum cost. A solution where no player can benefit by unilaterally changing their path is called a \emph{Nash Equilibrium}. Let $\T$ be a solution that minimizes the total cost paid. Note that $\T$ is a minimum Steiner tree for $G$. The \emph{price of stability} (\pos) is the ratio between the minimum cost of a Nash equilibrium and the cost of $\T$.

Let $p_{\T}(u,v)$ be the path in $\T$ between $u$ and $v$.
Let $v_1,\dots,v_n$ be the vertices of $\T$ in the order they appear in a depth first search of $\T$.
Let $MC$, the ``main cycle'', be the concatenation of $p_{\T}(v_1,v_2), p_{\T}(v_2,v_3),\dots, p_{\T}(v_{n-1},v_n), p_{\T}(v_n,v_1)$.
Note that each edge in $\T$ appears exactly twice in $MC$.
The following property will be helpful:

\begin{fact}\label{fct:MC-segment-contains-T-path}
  Any $x$ to $y$ path in $MC$ completely contains $p_{\T}(x,y)$.
\end{fact}

Define the class of edge $e$, $class(e)$, as $\alpha$ if $256^{\alpha} \le c(e) < 256^{\alpha+1}$. Without loss of generality, we assume that $c(e) \ge 1$ for all $e \in E$, so the minimum possible edge class is 0. 
For simplicity, define $\low{e} = 256^{class(e)}$, a lower bound for $c(e)$, and $\upp{e} = 256^{class(e)+1}$, an upper bound for $c(e)$. 

For each nonterminal $v$, let $\sigma_v$ be the minimum cost edge adjacent to $v$ in $G$.
Let $t_v$ be the terminal adjacent to $\sigma_v$.
Let $T^+$ be the extended optimal metric: $T^* \cup \{\sigma_v\}_{v \in V}$.
We maintain a dynamic set of nonterminals
\begin{equation*}
  \Z_S = \left\{ w \notin \T : c(\sigma_w) \le \low{e_w}/64 \right\}.
\end{equation*}
That is, $\Z_S$ are those nonterminals $w$ in solution $S$ whose first edge $e_w$ has cost within a constant factor of the cost of $\sigma_w$
For any $w \in S$, if $\sigma_w$ is added to $S$ while $w \in \Z_S$, then we will show that we will be able to pay for $\sigma_w$ if it remains in the final solution. We prove this fact in Section~\ref{sec:cost-analysis}. 
In the description of the algorithm, we denote the current state by $S_{curr}$. 
For ease of notation, we define $\Z = \Z_{S_{curr}}$.

The remaining definitions are modifications of key definitions from \cite{BiloFM13}. 
The interval around vertex $v \in \T$ with budget $y$, $I_{v,y}$, is the concatenation of its right and left intervals, $I^+_{v,y}$ and $I^-_{v,y}$, where $I^+_{v,y}$ is the maximal contiguous interval in $MC$ with $v$ a left endpoint such that
\begin{equation*}
  2\sum_{\alpha \ge 0} 256^{\alpha+1} H^2_{n_{I^+, \alpha}} \le y,
\end{equation*}
where $n_{I^+, \alpha}$ is the number of edges of class $\alpha$ in $I_{v,y}^+$ (repeated edges are counted every time they appear).
We define $I^-_{v,y}$ similarly.

The \emph{neighborhood of $v$ in state $S$}, $N_S(v)$ is an interval around $v$ as well as certain $w\not\in \T$ with $t_w$ in the interval. Formally,
\begin{equation*}
  N_S(v) =
  \begin{cases}
    I_{v,\frac{\low{e_v}}{56}} \cup \left\{ w \in \Z_S \middle| t_w \in  I_{v,\frac{\low{e_v}}{56}} \text{ and } c(\sigma_w) \le \frac{\low{e_v}}{64} \right\} &\text{  if } v\in\T, \\
    I_{t_v,\frac{\low{e_v}}{56}} \cup \left\{ w \in \Z_S \middle| t_w \in  I_{t_v,\frac{\low{e_v}}{56}} \text{ and } c(\sigma_w) \le \frac{\low{e_v}}{64} \right\} &\text{  otherwise.}
  \end{cases}
\end{equation*}
$N^+_S(v)$ and $N^-_S(v)$ are the right and left intervals of the neighborhood respectively (that is, the portions of $N_S(v)$ to the right and left of $v$ or $t_v$ respectively).
We denote $N_{S_{curr}}(v)$ as $N(v)$.
Roughly speaking, we are going to charge the cost of edges in the final solution not in $\T$ to the interval portions of non-overlapping right neighborhoods.

Observe that every edge in $N(v) \cap \T$ has class at most $class(e_v)-2$. If this were false,
\begin{equation*}
	2\sum_{\alpha>0} 256^{\alpha+1}H^2_{n_{I^+,\alpha}} \ge 256^{class(e_v)} > \frac{\low{e_v}}{56},
\end{equation*}
which contradicts the definition of $N(v)$.
A path $X = p_{\T}(x,y)$ is \emph{homogenous} if
\begin{equation*}
  |c_x(S) - c_y(S)| \le 4 \sum_{\alpha \ge 0} 256^{\alpha+1} H_{n_{X,\alpha}}^2.
\end{equation*}

If $X = p_{\T}(x,y) \subseteq N(v) \cap \T$ is a homogenous path then
\begin{equation*}
|c_x(S) - c_y(S)| \le 4 \sum_{\alpha \ge 0} 256^{\alpha+1} H^2_{n_{X, \alpha}}
				  \le 8 \sum_{\alpha \ge 0} 256^{\alpha+1} H^2_{n_{N^+(v), \alpha}}
				  \le \low{e} / 14.
\end{equation*}

$N(v)$ is homogenous if the following holds:
For all $x,y \in N(v)$ with $x,y \ne u_v$, a special vertex to be defined later, such that the path in $T^+$ from $x$ to $y$ does not contain $v$, $|c_x(S_{curr}) - c_y(S_{curr})| \le \frac{23\low{e_v}}{112}$.
Homogenous neighborhoods allow us to bound the difference in cost between any two vertices in $N(v)$ which will be useful when arguing that players have improving strategy changes.

\section{Algorithm}
\label{sec:algorithm}

The initial state of the algorithm is the minimum cost tree $\T$ connecting all the terminals to the root.
The algorithm carefully schedules a series of potential-reducing moves.
(Recall the potential function $\Phi(S) = \sum_{e \in E} c(e) H_{n_e(S)}$ introduced in Section~\ref{sec:prelims}).
Since there are finitely many states possible, such a series of moves must always be finite.
Since any improving move reduces potential, we must be at a Nash equilibrium if there is no potential reducing move.
These moves are scheduled such that if any edge outside of $\T$ is introduced, it is subsequently accounted for by charging to some part of $\T$.
In particular, we will show that at any point in the process, and therefore in the equilibrium state at the end, the total cost of these edges is bounded by $O(1) \cdot c(\T)$.

\begin{figure}
  \centering
  \includegraphics{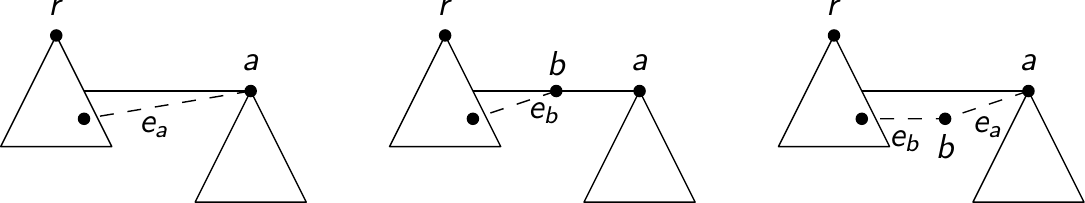}
  \caption{Types of critical improving moves. Dotted edges represent the new edges being added.}
  \label{fig:types-of-moves}
\end{figure}
The algorithm is a series of loops, which we run repeatedly until we reach a Nash equilibrium.
Each loop begins with a terminal, $a$, performing either a \emph{safe} improving move, or a \emph{critical} improving move.
In both cases, $a$ switches strategy to follow a new path to the root.
Let $S$ be the state before the start of the loop.
A safe improving move is one which results in some state $S' \subseteq \T \cup S$, i.e., the new path of $a$ contains edges currently in $S$ and edges in the optimal tree $T^*$.
A safe improving move requires no additional accounting on our part.
A critical improving move on the other hand introduces one or two new edges that must be accounted for (see Figure~\ref{fig:types-of-moves}).
We will show later that in any non-equilibrium state, a safe or critical improving move always exists (see Lemma~\ref{lem:final-state-is-ne}).

The algorithm will use a sequence of (potential-reducing) moves to account for the new edges introduced by a critical move.
At a high level, each of these edges is accounted for in the following way.
Let $e_v$ be the edge in question, and $v$ be the first vertex using $e_v$ on its path to the root.
\begin{enumerate}
  \item In some neighborhood around $v$, perform a sequence of moves to ensure that for every pair of vertices (excluding $v$ and at most one other special vertex), the difference in shared costs of these vertices is not too large. (Recall that the while nonterminals do not pay anything, the shared cost of a nonterminal $u$ is defined to be $c_u(S)$, the cost that a terminal using $u$ pays on its subpath from $u$ to the root). This sequence of moves must be potential-reducing, and cannot add any edges outside of $T^*\cup S$ to the solution.
  \item For every vertex $y$ in the neighborhood around $v$, $v$ has an alternative path to the root consisting of the path in $T^+$ to $y$, and $y$'s path to the root. (Recall from Section~\ref{sec:prelims} that $T^+$ is the optimal tree, $\T$, augmented with minimum cost edges incident on nonterminals $\{\sigma_w : w \text{ is a nonterminal}\}$.)
    \begin{enumerate}
      \item If there is a $y$ for which this alternative path is an improving path for $v$, then $v$ can switch to this new path and $e_v$ will be removed from the solution.
      \item If every path is \emph{not} improving for $v$, then we show that every vertex in the neighborhood of $v$ has an improving move that uses $e_v$.
    \end{enumerate}
\end{enumerate}

These steps ensure that we either remove $e_v$ from the solution, or else for any vertex $y$ in the neighborhood we remove edge $e_y \not\in T^*$ from the solution.
We elaborate on the steps above, referencing the subroutines described in Algorithm~\ref{alg:helper-functions} -- \func{Homogenize}, \func{Absorb}, and \func{MakeTree}:

\textbf{Step 1:} This is accomplished in two ways.
For any path in $\T$, the \func{Homogenize} subroutine ensures that a path in $T^*$ is homogenous.
Recall that this gives a bound (relative to the cost of $e_v$) on the difference in shared costs of the endpoints of the path.
Additionally, for any pair of adjacent vertices, if the difference in the shared costs is more than the cost of the edge between them, then one vertex must have an improving move through this edge.
This move adds no edges outside of $T^*$.
The second way of bounding differences in shared cost is much weaker, but we will use it only a small number of times.
Overall, the path between any two vertices in the neighborhood will comprise homogenous segments connected by edges whose cost is bounded by the second method above.
Adding up the cost bounds for each of these segments gives us the total bound.
Lemma~\ref{lem:cost-after-homogenize-loop} gives the technical details.

\textbf{Step 2(a):} The purpose of this step is to establish that either the shared cost of $v$ is not much larger than the shared cost of every other vertex in its neighborhood, or that we can otherwise remove $e_v$ from the solution.
If the shared cost of $v$ is much larger than some other vertex in the neighborhood, then it is also much larger than the shared cost of an adjacent vertex (call it $q$) in $T^+$. This is because every pair of vertices in the neighborhood have a similar shared cost (by Step 1).
Then, $v$ has a lower cost path to the root consisting of the $(v,q)$ edge, combined with $q$'s current path to the root.
Such a move would remove $e_v$ from the solution.

\textbf{Step 2(b):} If we reach this step, we need to account for the cost of $e_v$ by making every other vertex in the neighborhood give up its first edge, if that edge is not in $T^+$.
This ensures that at the end, the edges in the solution that are not in $T^+$ will be very far apart.
This is accomplished via the \func{Absorb} function: $v$ is currently paying the entire cost of $e_v$, while any vertex that would switch to using $v$'s path to the root would only pay at most half the cost of $e_v$.
Furthermore, if vertices close to $v$ in $T^+$ switch first, vertices farther from $v$ (who must pay a higher cost to buy a path to $v$) will reap the benefits of more sharing, and therefore a further reduction in shared cost.
This is formalized in the definition of \func{Absorb}.

There are some other details which we mention here before moving on to a more formal description of the algorithm:
\begin{itemize}
  \item If $v$ is a nonterminal, let $u_v$ be the terminal that added $v$ as part the critical move. We avoid including $u_v$ in any path provided to the \func{Homogenize} subroutine. This is because \func{Homogenize} switches the strategies of terminals to follow the strategy of some terminal on input path. If terminals were switched to follow $u_v$'s path, this would increase the sharing on $e_v$, when it is required at the beginning of Step (2b) that only one terminal is using $e_v$. When $v$ is a terminal, then $u_v$ is undefined and this problem does not exist. We define two versions of a loop of the algorithm, defined as \func{MainLoop} in Algorithm~\ref{alg:main-loop}, to account for this difference.
  \item We have only described how to account for a single edge, but sometimes a critical move adds two new edges that must be accounted for. Suppose $e_a$ and $e_b$ are the new edges added by $a$ ($a$ is a terminal and $b$ is a nonterminal). Then we run \func{MainLoop}$(e_b)$ first, and then \func{MainLoop}$(e_a)$. The first loop does not increase sharing on $e_a$, so the second loop is still valid.
  \item We assume the existence of a function $\func{MakeTree}$. This function takes as input a set of strategies. Its output is a new set of strategies such that (1) the new set of strategies has lower potential than the old set, (2) the edge set of the new strategies is a subset of the old edge set, and (3) the edge set of the new strategies is a tree.
  In particular, \func{MakeTree}$(S_{curr} \setminus \{p_{u_v}(S_{curr}),p_v(S_{curr})\})$, used on line~\ref{line:make-tree-steiner} does not increase sharing on $e_v$, since $v$ and $u_v$ are the only two vertices using $e_v$ on their path to the root.
  \func{MakeTree}$(S_{curr} \setminus \{p_{u_v}(S_{curr}),p_v(S_{curr})\})$ will also not increase sharing on $e_{u_v}$ if this edge has just been added (and therefore $u_v$ is the only vertex using the edge).
  We will not go into more detail about this function, since an identical function was used in both \cite{BiloFM13} and \cite{FiatKLOS06}.
  \item We assume that all edges in $E$ with $c(e) > c(\T)$ have been removed from the graph.
  This is without loss of generality:
  if the final state $S_f$ is a Nash equilibrium, then $S_f$ is still an equilibrium after reintroducing $e$ with $c(e) > c(\T)$.
  This is because any vertex with an improving move that adds such an edge $e$ also has a path to the root (the path in $\T$) with total cost less than $c(e)$.
\end{itemize}

We walk through the peusdocode next:
We execute the \func{MainLoop} function given in Algorithm~\ref{alg:main-loop} either once or twice, once for each edge not in $\T \cup S$ that is added by a critical move.
If two edges have been added, we execute in the order \func{MainLoop}$(e_b)$ then \func{MainLoop}$(e_a)$ (where $a$ is the terminal and $b$ is the nonterminal).
We define two versions of \func{MainLoop}$(e_v)$, one when $v$ is a terminal, and one when $v$ is a nonterminal, appearing on lines~\ref{line:main-loop-terminal} and \ref{line:main-loop-steiner} respectively.
When $v$ is a nonterminal, we denote the terminal which added $e_v$ to the solution as part of the initial improving move as $u_v$.
For brevity, we define $u_v$ as ``empty'' when $v$ is a terminal.
Thus if $v$ is a terminal, define $N(v) \setminus \{u_v\} = N(v)$.

\begin{algorithm}[h!]
  \caption{Main loop to be executed for each edge added to the solution as part of a critical move.}
  \label{alg:main-loop}
  \begin{algorithmic}[1]
    \Function{MainLoop}{$e_v$} \label{line:main-loop-steiner}
      \Comment{$v$ is a nonterminal and $u_v$ the terminal which added $e_v$ as part of a critical move.}
      \While{any of the following \textbf{if} conditions are true} \label{line:homogenize-loop-steiner}
        \If{$\exists X=p_{\T}(x,y) \in N(v) \cap \T$ with $u_v,v \not\in X$ and $X$ not homogenous} \xspace \Call{Homogenize}{$X$} \label{line:check-homogenous-steiner}
        \EndIf
        \If{$\exists x,y\in N(v) \setminus \left\{v\right\}$ adjacent to $u_v$ with $c_{x}(S_{curr}) - c_{y}(S_{curr}) > c(x,u_v) + c(u_v,y)$} \label{line:check-adjacent-to-u}
          \State Replace $x$'s strategy with $(x,u_v) \cup (u_v,y) \cup p_{y}(S_{curr})$.
        \EndIf
        \If{$\exists w \in N(v) \setminus \T$ such that $t_w \ne v,u_v$  with $\left|c_w(S_{curr}) - c_{t_w}(S_{curr})\right| > c(\sigma_w)$} \label{line:check-cost-w-tw-steiner}
          \State Assuming WLOG $c_{t_w}(S_{curr}) > c_w(S_{curr})$, replace $t_w$'s strategy with $\sigma_w \cup p_w(S_{curr})$. \label{line:switch-sigmav-steiner}
        \EndIf
        \If{$S_{curr} \setminus \{p_{u_v}(S_{curr}),p_v(S_{curr})\}$ is not a tree}
          \State \Call{MakeTree}{$S_{curr} \setminus \{p_{u_v}(S_{curr}),p_v(S_{curr})\}$} \label{line:make-tree-steiner}
        \EndIf
      \EndWhile
      \For{$q \in N(v) \setminus \left\{v,u_v\right\}$ adjacent in $T^+$ to either $v$ or $u_v$} \label{line:deleting-begin-steiner}
        \If{$c(v,q) + c_q(S_{curr}) < c_v(S_{curr})$} \label{line:deleting-condition-steiner}
          \State $v$ changes strategy to $(v,q)\cup p_q(S_{curr})$.
          \State \Return
        \EndIf
      \State Repeat the previous 3 lines substituting $u_v$ for $v$. \\
      \Comment{Note that $u_v$ changing strategy will remove $v$ from the solution.}
      \EndFor \label{line:deleting-end-steiner}
      \State \Call{Absorb}{$v$}
    \EndFunction 
    \item[]
      \Function{MainLoop}{$e_v$} \label{line:main-loop-terminal}
      \Comment{$v$ a terminal.}
      \While{any of the following \textbf{if} conditions are true} \label{line:homogenize-loop-terminal}
        \If{$\exists X=p_{\T}(x,y) \in N(v) \cap \T$ with $v \not\in X$ and $X$ is not homogenous} \xspace \Call{Homogenize}{$X$} \label{line:check-homogenous-terminal}
        \EndIf
        \If{$\exists w \in N(v) \setminus \T$ such that $t_w \ne v$ with $\left|c_w(S_{curr}) - c_{t_w}(S_{curr})\right| > c(\sigma_w)$}  \label{line:check-cost-w-tw-terminal}
          \State Assuming WLOG $c_{t_w}(S_{curr}) > c_w(S_{curr})$, replace $t_w$'s strategy with $\sigma_w \cup p_w(S_{curr})$.  \label{line:switch-sigmav-terminal}
        \EndIf
        \If{$S_{curr} \setminus \left\{p_v(S_{curr})\right\}$ is not a tree} \xspace \Call{MakeTree}{$S_{curr} \setminus \left\{p_v(S_{curr})\right\}$} \label{line:make-tree-terminal}
        \EndIf
      \EndWhile
      \For{$q \in N(v)$ adjacent in $T^+$ to $v$} \label{line:deleting-begin-terminal}
        \If{$c(v,q) + c_q(S_{curr}) < c_v(S_{curr})$} \label{line:deleting-condition-terminal}
          \State $v$ changes strategy to $(v,q)\cup p_q(S_{curr})$.
          \State \Return
        \EndIf
      \EndFor \label{line:deleting-end-terminal}
      \State \Call{Absorb}{$v$}
    \EndFunction 
  \end{algorithmic}
\end{algorithm}

\begin{algorithm}[h]
  \caption{Helper functions for Algorithm~\ref{alg:main-loop}.}
  \label{alg:helper-functions}
  \begin{algorithmic}[1]
	  \setcounter{ALG@line}{24}
    \Function{Homogenize}{$X=p_{\T}(x,y)$}
      \State Let $X = (x = x_1,x_2,\dots,x_k,x_{k+1} = y)$
      \State Let $S'$ be the current state.
      \For{$i \gets 1$ to $k$}
        \For{$j \gets i$ down to $1$}
          \State Change $x_j$'s strategy to $p_{\T}(x_j,x_{i+1}) \cup p_{x_i}(S)$.
        \EndFor
        \If{$\Phi(S_{curr}) < \Phi(S')$} \xspace \Return
        \Else \xspace Reset state to $S'$
        \EndIf
      \EndFor
    \EndFunction
    \item[]
    \Require $c_q(S) \ge c_v(S)-\frac{2 \cdot \low{e_v}}{7} \quad \forall q\in N(v)\setminus \left\{u_v\right\}$ \label{line:absorb-precondition}
    \Comment{See Lemma~\ref{lem:cost-after-deleting}} 
    \Function{Absorb}{$v$}
      \Comment{$v$ absorbs $N(v) \setminus \{u_v\}$}
      \For{$q \in N(v) \cap \T \setminus \left\{u_v\right\}$ in breadth-first order from $r$ according to $\T$}
        \If{$v \not\in \T$} \xspace Change $q$'s strategy along with its descendants to $p_{\T}(q,t_v) \cup \sigma_v \cup p_v(S)$. \label{line:absorb-T}
        \Else \xspace Change $q$'s strategy along with its descendants to $p_{\T}(q,v) \cup p_v(S)$. \label{line:absorb-Z}
        \EndIf
      \EndFor
      \State Let $S'$ be the current state.
      \For{$q \in N(v) \setminus \T$, in reverse breadth-first order from $r$ according to $S'$} 
        \State Change $q$'s strategy along with its descendants to $\sigma_q \cup p_{t_q}(S')$. \label{line:absorb-steiner}
      \EndFor
    \EndFunction
  \end{algorithmic}
\end{algorithm}

The \textbf{while} loops at lines~\ref{line:homogenize-loop-steiner} and \ref{line:homogenize-loop-terminal} terminate with $N(v)$ being homogenous.
For any violated \textbf{if} statement within the \textbf{while} loop, we perform a move that reduces potential, and does not increase sharing on $e_v$, or on $e_{u_v}$ if it was added along with $e_v$ as part of $u_v$'s critical move.
In Lemma~\ref{lem:cost-after-homogenize-loop} we show that if none of these \textbf{if} conditions hold, $N(v)$ is homogenous.
Therefore, this \textbf{while} loop eventually terminates in a homogenous state.

We next use the cost bound given by Lemma~\ref{lem:cost-after-homogenize-loop} to ensure that the cost that $v$ pays is similar to the cost every other vertex in $N(v)$ pays.
If these costs are not close, we show in Lemma~\ref{lem:cost-after-deleting} that the condition at line~\ref{line:deleting-condition-steiner}/\ref{line:deleting-condition-terminal} will be true, and $e_v$ will be deleted from the solution.

If $e_v$ is still present at this point, we finally call the \func{Absorb} function.
Lemma~\ref{lem:cost-after-deleting} ensures that the precondition of the \func{Absorb} function is met.
We use this condition to show that the switches made by all the vertices in $N(v)$ in the \func{Absorb} function are improving, and therefore reduce potential.

Note that although we do not make this explicit, if at any point $S_{curr}$ contains edges that are not part of $p_u(S_{curr})$ for any terminal $u$, these edges are deleted immediately. 
This ensures that any nonterminal in $S_{curr}$ is always used as part of some terminal's path to $r$.

\section{Analysis}
\label{sec:analysis}

In this section, we first prove some properties about the algorithm.
Then we analyze the cost of the final Nash equilibrium.

\subsection{Termination}
\label{sec:correctness}

We first show that all parts of the algorithm reduce potential, guaranteeing that the algorithm terminates (by the definition of the potential function, the minimum decrease in potential is bounded away from 0).

Most steps in the algorithm involve single terminals making improving moves, and therefore these steps reduce potential.
There are two parts of the algorithm for which it is not immediately obvious that potential is reduced: the \func{Homogenize} function and the \func{Absorb} function.
We first show that the \func{Homogenize} function reduces potential.

\begin{theorem}
  Suppose there is a path $X = p_{\T}(x,y) \in N(v)$ which is not homogenous.
  Let $(x = x_1,x_2,\dots,x_k,x_{k+1} = y)$ be the sequence of vertices in $X$.
  Then there exists a prefix of $X$, $(x_1,\dots,x_i)$, such that the sequence of moves in which each $x_j,$ $j \in \{ 1, \ldots, i \},$ switches its strategy to $p_{\T}(x_j,x_{i+1}) \cup p_{x_{i+1}}(S)$ reduces potential.
  \label{thm:prefix-switch-reduces-potential}
\end{theorem}

Note that the order in which the vertices move does not affect the change in potential of the entire sequence of moves. However, to help prove the theorem, we will assume that the vertices execute these moves in the order $x_i,x_{i-1},\dots,x_1$, the order given in Algorithm~\ref{alg:helper-functions}.
Let $e_j = (x_j,x_{j+1})$.
Let $S$ be the state before the prefix move starts, and let $S_j$ be the state just after $x_j$ switches its strategy.
Note that $S_{j+1}$ is the state just before $x_j$ switches.

\begin{lemma}
  The prefix move given in Theorem~\ref{thm:prefix-switch-reduces-potential} for prefix $(x_1,\dots,x_i)$ does not reduce potential only if
  \begin{equation*}
    \sum_{j=1}^i c_{x_j}(S) - c_{i+1}(S) \le 2 \sum_{j=1}^i 2 H_j c(e_j).
  \end{equation*}
  \label{lem:potential-reduction-condition}
\end{lemma}
\begin{proof}
  The change in potential caused by $x_j$'s switch is $c_{x_j}(S_{j}) - c_{x_j}(S_{j+1})$.
  Partition the edges of $x_j$'s strategy in $S$ into 3 sets: edges in $p_{\T}(x_j,x_{i+1})$, edges in $p_{x_{i+1}}(S)$, and all other edges, called $E_{1,j}$, $E_{2,j}$, and $E_{3,j}$ respectively. Additionally, let $E_{4,j}$ be the remaining edges in $x_{i+1}$'s strategy: those not in $E_{2,j}$. 
  For edge $e$, let $n_e(S)$ be the number of players using edge $e$ in state $S$.
  \begin{itemize}
    \item $c^{E_{1,j}}_{x_j}(S_{j+1}) \ge 0 = c^{E_{1,j}}_{x_j}(S) - \sum_{e \in E_{1,j}} \frac{1}{n_e(S)} c(e)$, since $c^{E_{1,j}}_{x_j}(S) = \sum_{e \in E_{1,j}} \frac{1}{n_e(S)} c(e)$. 
    \item $c^{E_{2,j}}_{x_j}(S_{j+1}) - c^{E_{2,j}}_{x_{i+1}}(S_{j+1}) = c^{E_{2,j}}_{x_j}(S) - c^{E_{2,j}}_{x_{i+1}}(S)=0$ since all edges in $E_2$ are shared.
    \item $c^{E_{3,j}}_{x_j}(S_{j+1}) \ge c^{E_{3,j}}_{x_j}(S)$ since $x_j$'s cost on this edge set has not decreased, since no sharing has been added.
    \item $c^{E_{4,j}}_{x_{i+1}}(S_{j+1}) \le c^{E_{4,j}}_{x_{i+1}}(S)$ since sharing on these edges only increases.
  \end{itemize}
  These facts give us
  \begin{align*}
    c_{x_{i+1}}(S_{j+1}) - c_{x_j}(S_{j+1}) &= c^{E_{4,j}}_{x_{i+1}}(S_{j+1}) + c^{E_{2,j}}_{x_{i+1}}(S_{j+1}) - c^{E_{1,j}}_{x_j}(S_{j+1}) - c^{E_{2,j}}_{x_j}(S_{j+1}) - c^{E_{3,j}}_{x_j}(S_{j+1}) \\
    \le c^{E_{4,j}}_{x_{i+1}}(S) &+ c^{E_{2,j}}_{x_{i+1}}(S) - c^{E_{2,j}}_{x_j}(S) - c^{E_{1,j}}_{x_j}(S) + \sum_{e \in E_{1,j}} \frac{1}{n_e(S)} c(e) - c^{E_{3,j}}_{x_j}(S) \\
    &=  c_{x_{i+1}}(S) - c_{x_j}(S) + \sum_{e \in E_{1,j}} \frac{1}{n_e(S)} c(e).
  \end{align*}
  Then,
  \begin{align*}
    c_{x_j}(S_{j}) - c_{x_j}(S_{j+1}) &\le \sum_{h=j}^i 2\frac{1}{h-j+1} c(e_h) + c_{x_{i+1}}(S_{j+1}) - c_{x_j}(S_{j+1}) \\
    &\le  c_{x_{i+1}}(S) - c_{x_j}(S) + \sum_{e \in E_{1,j}} \frac{1}{n_e(S)} c(e) + \sum_{h=j}^i 2\frac{1}{h-j+1} c(e_h).
  \end{align*}
  Summing over all $j$ gives
  \begin{align*}
    \sum_{j=1}^i \left( c_{x_j}(S_j) - c_{x_j}(S_{j+1})  \right)\ge 0 \implies \\ \sum_{j=1}^i \left( c_{x_j}(S) - c_{x_{i+1}}(S)  \right)&\le \sum_{j=1}^i 2\sum_{h=j}^i \frac{1}{h-j+1} c(e_h) + \sum_{j=1}^i \sum_{e \in E_{1,j}} \frac{1}{n_e(S)} c(e) \\
    &\le 2\sum_{j=1}^i H_j c(e_j) + \sum_{j=1}^i c(e_j) \le 4\sum_{j=1}^i H_j c(e_j). \qedhere
  \end{align*}
\end{proof}

Assume that for all $i$, $1\le i \le k$, the prefix move given in Theorem~\ref{thm:prefix-switch-reduces-potential} for prefix $(x_1,\dots,x_i)$ does not reduce potential.
Then we show that $p_{\T}(x,y)$ is homogenous.
Let 
\begin{equation*}
  g_j =
  \begin{cases}
    \frac{1}{j(j+1)} &\text{ for } j < k \text{,  and} \\
    \frac{1}{j} &\text{ for } j = k. \\
  \end{cases}
\end{equation*}
Note that $\sum_{j=i}^k g_j = 1/i$ for all $i \le k$.
From Lemma~\ref{lem:potential-reduction-condition} we have 
\begin{align*}
  &\sum_{j=1}^i c_{x_j}(S) - c_{x_{i+1}}(S) \le 2 \sum_{j=1}^i 2 H_j c(e_j) \quad \forall i, 1\le i \le k \\
  \implies &\sum_{i=1}^k \left[ g_i \sum_{j=1}^i c_{x_j}(S) - c_{x_{i+1}}(S) \right] \le \sum_{i=1}^k 2 g_i \sum_{j=1}^i 2 H_j c(e_j).
\end{align*}
Rewriting $\sum_{j=1}^i c_{x_j}(S) - c_{x_{i+1}}(S)$ as $\sum_{j=1}^i \sum_{h=j}^i c_{x_h}(S) - c_{x_{h+1}}(S)$ and rearranging, we obtain
\begin{align*}
  &\sum_{i=1}^k \left[ g_i \sum_{j=1}^i j \left( c_{x_j}(S) - c_{x_{j+1}}(S)  \right)\right] \le \sum_{i=1}^k 2 g_i \sum_{j=1}^i 2 H_j c(e_j).
\end{align*}
Rearranging the sums, we get 
\begin{align*}
  &\sum_{j=1}^k \left[ j \left( c_{x_j}(S) - c_{x_{j+1}}(S) \right) \sum_{i=j}^k g_i \right] \le 4\sum_{j=1}^k H_j c(e_j) \sum_{i=j}^k g_i \\
  \implies &\sum_{j=1}^k c_{x_j}(S) - c_{x_{j+1}}(S) \le 4\sum_{j=1}^k \frac{H_j}{j} c(e_j) \\
  \implies &c_{x}(S) - c_{y}(S) \le 4\sum_{\alpha\ge 0} \left[ 256^{\alpha+1} \sum_{j=1}^{n_{\alpha,X}} \frac{H_j}{j}\right] \le  4\sum_{\alpha\ge 0} 256^{\alpha+1}  H^2_{n_{\alpha,X}}.
\end{align*}
This completes the proof of Theorem~\ref{thm:prefix-switch-reduces-potential}.

Before we can analyze the \func{absorb} function, we need to show that the precondition for the \func{Absorb} function (see just before line~33) is satisfied.
We state and prove the precondition in Lemma~\ref{lem:cost-after-deleting}.
The proof of Lemma~\ref{lem:cost-after-deleting} requires a homogenous solution, so we first prove homogeneity in the follow lemma (Lemma~\ref{lem:cost-after-homogenize-loop}).
\begin{lemma}
  \label{lem:cost-after-homogenize-loop}
  When the \textbf{while} loops on lines~\ref{line:homogenize-loop-steiner} and \ref{line:homogenize-loop-terminal} terminate, $N(v)$ is homogenous.
\end{lemma}
\begin{proof}
  Let $S = S_{curr}$.
  Let $x,y \in N(v)$ such that the path in $T^+$ from $x$ to $y$ does not contain $v$.
  Let $P$ be this path.
  In the simplest case, both $x$ and $y$ are in $\T$, and $u_v$, if it exists, does not lie in $P$.
  Then, we can bound the cost difference between $x$ and $y$ using the homogenous property, guaranteed by line~\ref{line:check-homogenous-steiner}/\ref{line:check-homogenous-terminal} of the \func{MainLoop} function.
  That is, $|c_x(S) - c_y(S)| \le \frac{\low{e_v}}{14}$.
  Next let us consider the general case, in which both $x$ and $y$ are nonterminals, and $u_v\in P$.
  Let $P = (x,t_x,\dots,u_l,u_v,u_r,\dots,t_y,y)$.
  We have the following bounds:
  \begin{equation*}
	|c_x(S) - c_{t_x}(S)| \le \frac{\low{e_v}}{64}, \qquad |c_{t_x}(S) - c_{u_l}(S)| \le \frac{\low{e_v}}{14}, \qquad  |c_{u_l}(S) - c_{u_r}(S)| \le \frac{2\low{e_v}}{256}
	\end{equation*}
	\begin{equation*}
	|c_{t_y}(S) - c_{u_r}(S)| \le \frac{\low{e_v}}{14}, \qquad |c_y(S) - c_{t_y}(S)| \le \frac{\low{e_v}}{64}.
  \end{equation*}
  Combining these ensures $|c_x(S) - c_y(S)| \le \frac{23\low{e_v}}{112}$.
\end{proof}

We use Lemma~\ref{lem:cost-after-homogenize-loop} to show the precondition for the \func{Absorb} function.
\begin{lemma}\label{lem:cost-after-deleting}
	If the \textbf{for} loops on line~\ref{line:deleting-begin-steiner}/\ref{line:deleting-begin-terminal} terminate without the algorithm returning, then
	\begin{equation*}
		c_p(S_{curr}) \ge c_v(S_{curr})-\frac{2 \cdot \low{e_v}}{7}
  \end{equation*}
  for all $p \in N(v), p \ne u_v$.
\end{lemma}

\begin{proof}
  Let $S = S_{curr}$.
  Again, we suppose that $u_v$ is not adjacent to $v$ in $\T$.
  Suppose the claim does not hold, that is, $c_v(S) - c_p(S) > \frac{2 \cdot \low{e_v}}{7}$ for some $p \in N(v)$
  Let $P = (p,\dots,q,v)$ be the path from $p$ to $v$ in $T^+$ and suppose $q \not= u_v$.
  By Lemma~\ref{lem:cost-after-homogenize-loop}, 
  \begin{equation*}
    |c_p(S) - c_q(S)| \le \frac{23\low{e_v}}{112} \implies c_v(S) - c_q(S) > \frac{9\low{e_v}}{112}.
  \end{equation*}
  But, the cost of edge $(v,q)$ is at most $\frac{\low{e_v}}{256} < \frac{9\low{e_v}}{112}$, so $v$ must have had an improving move in line~\ref{line:deleting-condition-steiner}/\ref{line:deleting-condition-terminal}.
  
  Now suppose $u_v$ is adjacent to $v$ in $\T$ and lies in $P$. Then $P = (p, \ldots, q, u_v, v)$.
  By Lemma~\ref{lem:cost-after-homogenize-loop}, 
  \begin{equation*}
    |c_p(S) - c_q(S)| \le \frac{23\low{e_v}}{112} \implies c_v(S) - c_q(S) > \frac{9\low{e_v}}{112}.
  \end{equation*}
  But, the cost of edge $(v,u_v)$ is at most $\frac{\low{e_v}}{256}$ (since this edge is in $\T$), and the same is true of the edge $(u_v,q)$. Therefore $v$ can switch its path to $(v,u_v) \cup (u_v,q) \cup p_q(S)$ and pay cost at most $2\frac{\low{e_v}}{256} + c_q(S) < \frac{9\low{e_v}}{112} + c_q(S)<c_v(S)$. Note that this improving move for $v$ implies an improving move for $u_v$ in line~\ref{line:deleting-condition-steiner} where $p_{u_v}(S)$ no longer includes $v$ but instead consists of $(u_v,q) \cup p_q(S)$. 
\end{proof}

Now that we have shown the precondition is satisfied, we must show that the \func{absorb} function reduces potential.

\begin{theorem} \label{thm:absorb-reduces-potential}
	If $c_q(S_{curr}) \ge c_v(S_{curr}) - \frac{2 \cdot \low{e_v}}{7}$ for all $q \ne u_v \in N(v)$, then every strategy change in \func{Absorb} reduces potential.
\end{theorem}
The proof follows from the following three lemmas.

\begin{lemma}
  \label{lem:no-sharing-ev}
  At the beginning of the \func{Absorb}$(v)$ function, $v$ and $u_v$ (if $v$ is a nonterminal) are the only vertices using $e_v$.
\end{lemma}

\begin{proof}
	First consider the case where $v$ is a nonterminal. We show that no edge added to the solution by the critical move (this is at least $e_v$, and possibly $e_{u_v}$ as well) is used by any terminal other than $u_v$ as a result of lines~\ref{line:main-loop-steiner} to~\ref{line:deleting-end-steiner}. At the start of \func{MainLoop} on line~\ref{line:main-loop-steiner}, $u_v$ is the only terminal using $e_v$ and $e_{u_v}$, by the definition of a critical move. 
	
	Note that \func{Homogenize}$(X)$ only increases sharing on edges in $X$ and the path $p_{x}(S)$ for some $x \in X$. However, since in line~\ref{line:check-homogenous-steiner} we only consider segments not containing $v$, no edge adjacent to $v$ is in $X$ either. And since $v$ is the only vertex using $e_v$ and $v \notin X$, there is no $x$ for which $e_v \in p_x(S)$. Therefore line~\ref{line:check-homogenous-steiner} does not increase sharing on $e_v$. An identical argument shows that this line does not increase sharing on $e_{u_v}$, if $v$ is a nonterminal.
	
	Next consider line~\ref{line:check-adjacent-to-u} of Algorithm~\ref{alg:main-loop}. The only edges on which sharing can increase are $(x,u_v)$, $(u_v,y)$, and edges in $p_y(S)$. By our simplifying assumption, $u_v$ is not adjacent to $v$ in $\T$ and therefore $x,y \not= v$. Therefore neither $(x,u_v)$ nor $(u_v,y)$ is equal to $e_{u_v}$ or $e_v$. Additionally, $e_{u_v}, e_v \notin p_y(S)$. Therefore sharing does not increase on either $e_{u_v}$ or $e_v$.
	
	Line~\ref{line:switch-sigmav-steiner} increases sharing on exactly one of $p_w(S)$ and $p_{t_w}(S)$, as well as $\sigma_w$, but only for $w \not= v, u_v$. We know that $u_v$ is the only terminal using $e_{u_v}$ or $e_v$ in $S$, so neither of these outcomes increases sharing on $e_v$. Line~\ref{line:make-tree-steiner} uses only the \func{MakeTree} function, which, by definition, does not increase the sharing on $e_{u_v}$ or $e_v$. 
	
	Lastly consider the \textbf{for} loop beginning at line~\ref{line:deleting-begin-steiner}. By definition, \func{Absorb}$(v)$ is only run if the \textbf{for} loop does not result in a change to $S_{curr}$, so this can not increase sharing on $e_v/e_{u_v}$. 
	
	If the condition in line~\ref{line:deleting-condition-steiner} is ever true, then the algorithm does not begin the \func{Absorb} function. Therefore we only need to consider what happens when the condition in line~\ref{line:deleting-condition-steiner} is never true, in which case the entire \textbf{for} loop has no effect. 
	
	Now consider the case where $v$ is a terminal. Note that before the start of the \func{MainLoop} function beginning at Line~\ref{line:main-loop-terminal}, $v$ is the only vertex using $e_v$ (either because $v$ only added a single new edge for the critical move, or $v$ added two edges but the execution of \func{MainLoop} on the second edge did not increase sharing on $e_v$ by the argument above).
Using the same argument as for the case of $v$ being a nonterminal, we can prove that this function does not increase sharing on $e_v$, and we omit the details to avoid repetition.
\end{proof}

\begin{lemma}
\label{lem:ab1} 
	Let $\{q_1, \ldots, q_k \}$ denote all vertices in $N(v)$ sorted by breadth first order from $v$ according to $\T$. Then changing $q_i$'s strategy as in Lines~\ref{line:absorb-T} and \ref{line:absorb-Z} is potential decreasing for all $i \in \{ 1 \ldots k \}$.
\end{lemma}

\begin{proof}
Suppose that $e_v = \{ v,w \}$. Let $S$ be the solution before \func{Absorb}$(v)$ is called and let $S^{i}$ be the solution after $q_i$ has changed strategy. Therefore $c_{q_i}(S^{i-1})$ is the cost paid by $q_i$ directly before switching, and $c_{q_i}(S^i)$ is the cost paid by $q_i$ directly after switching. $c_{q_i}(S^{i-1})$ is exactly equal to the cost paid by $q_i$ before any players switched, minus the reduction in $q_i$'s cost due to sharing on $p_{q_i}(S)$ from $q_1, \ldots, q_{i-1}$ changing their strategies. Since we know that $e_v \notin p_{q_i}(S)$, we can divide this reduction into two components: the reduction due to sharing on edges in $N(v) \cap \T \cap p_{q_i}(S)$, and the reduction due to sharing on edges in $p_w(S)\cap p_{q_i}(S)$. Denote the latter quantity by $dec_w$. The former quantity is upper bounded by the maximum cost $q_i$ could pay on $N(v) \cap \T$ in $S$ (remembering that $u_v$ may not be using any edges in $N(v)$ and thus not contributing to sharing), $2\frac{\low{e_v}}{56} + 2\frac{\low{e_v}}{256}$. So we can upper bound the total decrease in $q_i$'s cost due to the players' switching by $\frac{\low{e_v}}{28} + dec_w + \frac{\low{e_v}}{128}$. Therefore
	\begin{equation*}
		c_{q_i}(S^{i-1}) \ge c_v(S)-\frac{2 \cdot \low{e_v}}{7}-\frac{\low{e_v}}{28} -dec_w -\frac{\low{e_v}}{128}.
	\end{equation*}

  Suppose that $v$ is a terminal and consider the cost paid by $q_i$ directly after switching, $c_{q_i}(S^i)$. $q_i$ is sharing edge $e_v$ with at least one other player (namely $v$), so on the edges shared with $v$, $q_i$ is paying at most $c_v(S) -\frac{\low{e_v}}{2} - dec_w$. And on the edges not shared with $v$, namely the edges in $p_{\T}(q_i,v)$, $q_i$ pays at most $ 2\sum_{\alpha \ge 0} 256^{\alpha + 1} H_{n_{N^+(v), \alpha}} + 2 \frac{\low{e_v}}{256} \le \frac{\low{e_v}}{56} + \frac{\low{e_v}}{128}$. So 
	\begin{equation}
		c_{q_i}(S^i) \le c_v(S) - \frac{\low{e_v}}{2} - dec_w + \frac{\low{e_v}}{56} + \frac{\low{e_v}}{128} < c_v(S) - dec_w - \frac{6 \cdot \low{e_v}}{14}  < c_{q_i}(S^{i-1}). \label{eqn:absorbing}
\end{equation}
	Therefore, it is an improving move for $q_i$ to switch. 

  Now suppose that $v$ is a nonterminal. Then by definition, there is some terminal $u$ using the edges $\{u,v\}$ and $e_v$. Suppose that $q_i \not= q_1$. Then when $q_i$ switches in the absorbing process, it shares the cost of $e_v$ with $q_1$. For all $i \ge 2$, $q_i$ shares the cost of edge $e_v$ with (at least) $q_1$. Therefore, Equation \ref{eqn:absorbing} holds for all $i$. 
  
  The last case is when $v$ is a nonterminal and $q_1 = q_i$. But if this is the case then edge $\{ u,v \} \in \T$. If this were not true, then there is some path $p_\T(u,v) = \{ u, x_1, x_2, \ldots, x_n, x_{n+1} = v \}$. Since the network is quasi-bipartite, $x_n$ must be a terminal, and $x_n$ comes before $q_i$ in a breadth first traversal of $\T$ rooted at $v$, contradicting that $q_i=q_1$. Therefore $q_i$ is already using strategy $\{ u,v \} \cup p_v(S)$, so there is no switch to be done by $q_1$ in step 1 of the \func{Absorb} function. For $i \ge 2$, $q_i$ pays at most half the cost of edge $e_v$ (since it is shared with $u$), so Equation \ref{eqn:absorbing} holds.
\end{proof}

\begin{lemma}
\label{lem:ab2}
	Let $\{s_1, \ldots, s_k \}$ denote all nonterminals in $N(v) \setminus \T$ with $class(e_{s_i}) \le class(e_v)$, sorted in breadth-first order from $r$ according to $S$. Then for all $i \in \{ 1 \ldots k \}$, it is an improving move for $s_i$ to switch as in Line~\ref{line:absorb-steiner}. Moreover, after \func{Absorb}$(v)$ is completed, $e_{s_i}$ is no longer in the solution and is replaced by $\sigma_{s_i}$ (that is, $\sigma_{s_i}$ is the first edge on $s_i$'s path to the root, $p_{s_i}$).
\end{lemma}

\begin{proof}
  Let $S$ be the solution before \func{Absorb}$(v)$ is called. Let $s_i \in N(v) \setminus \T$. Consider first the case where $e_{s_i}=\sigma_{s_i}$. In this case $s_i$ is already taking strategy $e_{s_i} \cup p_{t_{s_i}}(S')$, since descendants are absorbed in lines~\ref{line:absorb-T} and~\ref{line:absorb-Z}, and there is no change to make in line~\ref{line:absorb-steiner}. Suppose for the rest of the proof that $e_{s_i}\not= \sigma_{s_i}$. 
  
  Let $S^i$ be the solution after $s_i$ has switched strategy. In particular, $c_{s_i}(S^{i+1})$ is the cost paid by $s_i$ directly before switching, and $c_{s_i}(S^i)$ is the cost paid by $s_i$ directly after switching.  $c_{s_i}(S^{i+1})$ is the cost paid by $s_i$ before the absorbing process began, $c_{s_i}(S)$, minus the cost reduction due to sharing on $p_{s_i}(S)$ due to other vertices switching as part of the absorbing process. Since $\sigma_{s_i} \notin p_{s_i}(S)$, there is no contribution to the latter term due to sharing on $\sigma_{s_i}$. The only other edges in $p_{s_i}(S)$ that can have increased sharing as a result of previous moves in the absorbing process are edges in $N(v) \cap \T \cap p_{s_i}(S)$ and edges in $p_w(S) \cap p_{s_i}(S)$.
  Therefore we have the same bound on this term as in Lemma \ref{lem:ab1}, $\frac{\low{e_v}}{28} + \frac{\low{e_v}}{128} + dec_w$. We also know, from the precondition for \func{Absorb}, that 
	\begin{equation*}
		c_{s_i}(S) \ge c_v(S) - \frac{2 \cdot \low{e_v}}{7}.
	\end{equation*}
	
Therefore the cost that $u$ pays immediately before switching, $c_{s_i}(S^{i+1})$, satisfies
\begin{equation*}
	c_{s_i}(S^{i+1}) \ge c_v(S) - \frac{2 \cdot \low{e_v}}{7} - \frac{\low{e_v}}{28} - \frac{\low{e_v}}{128} - dec_w = c_v(S) - dec_w - \frac{10 \cdot \low{e_v}}{28}.
\end{equation*}
	
Now consider $s_i$'s cost immediately after switching (before any descendants switch), $c_{s_i}(S^i)$. $s_i$ pays at most the entire cost of $\sigma_{s_i}$. As in the proof of Lemma \ref{lem:ab1}, $s_i$ pays at most $c_v(S) -\frac{\low{e_v}}{2} - dec_w$ on edges shared with $v$, and at most $\frac{\low{e_v}}{56} + \frac{\low{e_v}}{128}$ on edges between $t_{s_i}$ and $v$. 
So $c_{s_i}(S^i)$ satisfies
\begin{align*}
	c_{s_i}(S^i) &\le c(\sigma_{s_i}) + c_v(S) - \frac{\low{e_v}}{2} - dec_w + \frac{\low{e_v}}{28} + \frac{\low{e_v}}{128}\\
	&\le \frac{\low{e_v}}{64} + c_v(S) - \frac{\low{e_v}}{2} - dec_w + \frac{\low{e_v}}{14}\\
	&< c_v(S) - dec_w - \frac{11 \cdot \low{e_v}}{28} < c_{s_i}(S^{i+1}),
\end{align*}
which proves the first part of the lemma. The second part follows from the definition of the move.
\end{proof}

Theorems~\ref{thm:prefix-switch-reduces-potential} and~\ref{thm:absorb-reduces-potential} prove that the entire main loop is potential reducing. Since the minimum decrease in potential is bounded away from zero, and the potential is always at least zero, the algorithm necessarily terminates.

However, termination alone does not guarantee that the final state is a Nash equilibrium.
Since we have restricted the set of moves that the algorithm can perform, we must show that whenever an improving move is available to some terminal, there is also an improving move that is either a safe or critical move.
\begin{lemma} \label{lem:final-state-is-ne}
	The final state reached by the algorithm, $S_f$, is a Nash equilibrium.
\end{lemma}

\begin{proof}
	Suppose for contradiction that $S_f$ is not a Nash equilibrium, that there is an improving deviation for some player $q$. Consider the most improving deviation (that with lowest cost) and denote this lowest cost path to the root as $p_q(S'_f) = \{ q = q_1, q_2, \ldots, q_k = r \}$. Of all vertices $q_i$, consider the terminal with highest index such that $p_{q_i}(S_f) \not= \{q_i, q_{i+1}, \ldots , q_k=r \}$. Since $p_q(S'_f)$ is of lower cost to $q$ than the path $\{q=q_1, q_2, \ldots, q_{i-1}, q_i \} \cup p_{q_i}(S_f)$, it must also be the case that $\{q_i, q_{i+1}, \ldots , q_k=r \}$ is of lower cost to $q_i$ than $p_{q_i}(S_f)$. And by the maximality of index $i$, $q_i$ has an improving move where she can add edge $\{ q_i, q_{i+1} \}$ (if $q_{i+1}$ is a terminal), or edges $\{ q_i, q_{i+1} \}, \{ q_{i+1}, q_{i+2} \}$ (if $q_{i+1}$ is not a terminal). This is necessarily either a safe or critical move.
	
Therefore, if an improving move exists for any player at state $S_f$, then a safe or critical move exists for some player, contradicting termination of the algorithm.
\end{proof}

\subsection{Cost Analysis}\label{sec:cost-analysis}

Our goal for this section is to show our main result, Theorem~\ref{thm:main}.
We will show that $c(S_f) = O(c(T^*))$.
That is, we will show that that the cost of the final Nash equilibrium reached by the algorithm, $S_f$, is within a constant factor of the cost of the optimal tree, $\T$.

To establish the theorem, it is sufficient to show that $c(S_f \setminus T^*) = O(c(T^*))$. We devise a charging scheme that distributes the cost of edges in $S_f \setminus T^*$ among edges in $T^*$. Each $e \in S_f \setminus T^*$ must be an $e_v$ edge for some vertex $v$.  Furthermore, these $e_v$ edges were not later removed as the result of an absorbing process initiated from another $e_{v'}$. At a high level, this allows us to distribute the cost of each $e_v$ to the edges in the neighborhood $N(v) \cap \T$, since the \func{Absorb}$(v)$ function removes many other $e_{v'}$ edges where $v' \in N(v)$ from the solution. When $v$ is a terminal, this is the same argument used in \cite{BiloFM13}; however, we will need to take special care when distributing cost for $e_v$ when $v$ is a nonterminal as well as for some $\sigma_v$ edges when $v \notin \T$.
 
We first consider a set of edges that we will not charge to their neighborhood. Define
\begin{equation*}
E_\sigma = \{ e_v \in S_f | v \text{ is a nonterminal, } \frac{\low{e_v}}{64} \le c(\sigma_v) \}.
\end{equation*}
We bound the cost of $E_\sigma$ by the cost of edges in $S_f \setminus E_\sigma$.
 
 \begin{lemma} 
 \label{lem:shortedge} 
 $c(E_\sigma) = O(c(S_f \setminus E_\sigma))$. 
 \end{lemma}  
 
 \begin{proof} 
Let $e_v = (v,w) \in E_\sigma$, where $v$ is a nonterminal. First observe that since $(v,w) \in S_f$, $v$ is not a leaf. Let $e_u = (u,v)$, where $u$ is necessarily a terminal. Since $c(e_v) \leq 64 \cdot c(\sigma_v)$, we have that $c(e_v) \leq 64 \cdot c(e_u)$ from the definition of $\sigma_v$. Thus, we charge $c(e_v)$ to $c(e_u)$. Observe that no edge in $S_f \setminus E_\sigma$ is charged more than once since every nonterminal has a unique parent in $S_f$, and edges in $S_f \setminus E_\sigma$ are only charged the cost of the first edge used by their parent (if at all). 
\end{proof} 

Our goal now is to find a set of edges $e_v$ such that the right neighborhoods associated with edges of the same class are not overlapping.
In the absence of nonterminals, this is simple:
For every edge in $S_f \setminus \T$, the right neighborhoods of vertices corresponding to edges of the same class being overlapping implies that each edge is contained in the other's neighborhood.
Therefore, we argue that the second edge to arrive would have deleted the first through the \func{Absorb} function, which gives a contradiction.
With nonterminals, the same property does not hold.
When edge $e_v$ is added for some nonterminal $v$, $e_{u_v}$ will not be deleted from the solution, even if $u_v$ falls in $v$'s neighborhood.
The presence of $\sigma_v$ for which no \func{MainLoop}$(\sigma_v)$ was run (added, e.g., in line~\ref{line:absorb-T}) further complicates things.
To show that no right neighborhoods overlap, we will therefore remove some edges from $S_f \setminus (\T \cup E_{\sigma})$.

For nonterminal $v$, if $v$ is adjacent to at least two edges in $S_f \setminus (\T \cup E_{\sigma})$ and $\sigma_v$ is one such edge, remove $\sigma_v$ and charge it to one of the remaining edges adjacent to $v$.
Next, for any pair of edges $e_u$ and $e_v$ in $S_f \setminus (\T \cup E_{\sigma})$ such that $u$ was the terminal which added $e_v$, we delete the smaller of $e_u$ and $e_v$ and charge it to the remaining edge.
We are left with a set of edges which we denote $E^*$, each of which has been charged by at most two edges that were removed (and each edge removed is charged to some edge in $E^*$).

Our argument will charge to each edge in $\T$ at most one edge in $E^*$ of each class.
To make the argument simpler, it is desirable to charge those $\sigma_v$'s for which \func{MainLoop}$(\sigma_v)$ was never run to higher classes than their actual classes.
To this end, we increase the cost of each such $\sigma_v$ to $c(e_{\sigma_v})$, the cost of the first edge on $v$'s path in the state just before $\sigma_v$ was added.

\begin{lemma}
  \label{lem:overlap}
  For edges $e_u,e_v \in E^*$, if $class(e_v) = class(e_u)$, then $N^+(v)$ and $N^+(u)$ are disjoint.
\end{lemma}
\begin{proof}
  If $class(e_v) = class(e_u)$ and $N^+(v)$ and $N^+(u)$ overlap, then the vertex corresponding to the edge arriving first is in the neighborhood of the other.
  Suppose $N^+(v)$ and $N^+(u)$ overlap, $e_u$ preceded $e_v$, and $u \not\in N(v)$.
  If $u\not\in \T$, $c(\sigma_u) \le \frac{\low{e_u}}{64} = \frac{\low{e_v}}{64}$, and thus $t_u \in N(v)$ implies $u \in N(v)$.
  Then $t_u$ (or $u$ if $u \in \T$) must lie to the left of $N(v)$ in MC, which means that $N^+(u)$ strictly contains $N^-(v)$, a contradiction given that $class(e_v) = class(e_u)$.

  Suppose $e_u,e_v \in E^*$ with $class(e_v) = class(e_u)$.
  Suppose $e_u$ preceded $e_v$ and $N^+(v)$ and $N^+(u)$ overlap, which implies $u \in N(v)$ as shown above.
  We consider two cases for edge $e_v$, and derive contradictions in all cases.
  First, suppose \func{Mainloop}$(e_v)$ was run when $e_v$ was added.
  Then $e_u$ cannot exist after the completion of \func{Absorb}$(v)$, by definition of \func{Absorb}.
  
  If the \func{MainLoop}$(e_v)$ was not run, then $e_v = \sigma_w$ for some $w$, and $e_v$ was introduced because $w$ was in $\Z$.
  Let $e_{\sigma_w}$ be the first edge on $w$'s path to the root in the state just after $\sigma_w$ was added, which we denote $S_{\sigma_w}$ (this may be different from $w$'s current first edge, $e_w$).
  Then \func{MainLoop}$(e_{\sigma_w})$ caused the deletion of all edges $e_q$ such that $q \in N_{S_{\sigma_w}}(w)$.
  We consider two possible times when $e_u$ was added:
  If $e_u$ was added before $e_{\sigma_w}$, then \func{MainLoop}$(e_{\sigma_w})$ deleted $e_u$ since $u \in N_{S_{\sigma_w}}(w)$ because $class(e_u) = class(e_{\sigma_w})$.
  If $e_u$ was added after $e_{\sigma_w}$, $class(e_u) = class(e_{\sigma_w})$ implies that $w \in N(u)$.
  But then $e_{\sigma_w}$ was removed from the solution and $w$ was removed from $\Z$, so the addition of $\sigma_w = e_v$ was not possible.
\end{proof}

Given Lemma~\ref{lem:overlap}, the scheme from \cite{BiloFM13}  for distributing the cost of each $e_v$ to its neighborhood can be applied directly.
This gives us the following lemma, which along with Lemma \ref{lem:shortedge} establishes Theorem \ref{thm:main}.

\begin{lemma} \label{lem:charge-to-neighborhood}
The cost of each  $e_v \in E^*$ can be distributed to the edges in $N^+(v)$ (and its boundary) such that the total charge on any edge $e' \in T^*$ is $O(c(e'))$. 
\end{lemma}

\begin{proof}
Let $e_v \in E^*$. Let $\alpha = class(e_v)$. Throughout this proof we will be interested only in edges $N^+(v) \cap \T$. For simplicity, we will simply write $N^+(v)$ instead of $N^+(v) \cap \T$.

We first consider the case where $N^+(v) =MC=N(v)$. If $v$ is a terminal, then $e_v$ is the only edge not in $S_f \setminus E_\sigma$, since all other terminals are following their path in $\T$ to $v$. But all edges added by a critical move must have $c(e) \le c(\T)$ by definition. So in this case, $c(S_f \setminus E_\sigma) \le c(\T)$. If $v$ is a Steiner vertex, then $S_f \setminus E_\sigma$ can also include $\sigma_v$, since all terminals $u$ are using strategy $p_{\T}(u,t_v) \cup \sigma_v \cup p_v(S)$. Since $c(\sigma_v) \le c(e_v)$, we have $c(S_f \setminus E_\sigma) \le 2c(\T)$.

Now we consider the case where $N^+(v) \not= MC$. Recall that every edge in $N^+(v)$ has class at most $\alpha-2$. Let $a \not\in N^+(v)$ be the first edge to the right of $N^+(v)$ in $\T$, and let $\mu = class(a)$. There are two cases.

{\bf Case 1:} $\mu \ge \alpha -1$. In this case we can charge $c(e_v)$ to $a$. We show that only one edge of each class will get charged to $a$. Suppose that this is not the case, that there is some $e_{v'}$ with $class(e_{v'})=\alpha$, such that $a$ is the first edge to the right of $N^+(v')$ in $\T$. Then the edges of $N^+(v)$ have non-empty intersection with those of $N^+(v')$, contradicting Lemma \ref{lem:overlap}.

So the total cost charged to edge $a$ in this fashion is no greater than $\sum_{\gamma=0}^{\mu+1} 256^{\gamma+1} < 256^{\mu+3} \le 256^3c(a)$. Since each $e \in \T$ appears at most twice in $MC$, $e$ is charged a total cost of at most $2 \cdot256^3c(e)$.

{\bf Case 2:} $\mu \le \alpha-2$. In this case we charge $c(e_v)$ to a subset of the edges in $N^+(v)$. We first prove a technical claim.

\begin{claim}\label{claim:heavy-class}
	There exists some class $1 \le \beta \le \alpha-2$ such that 
	\begin{equation*}
		\frac{256^{\beta+1}H^2_{n_{N^+(v),\beta}}}{256^{\alpha-1}} \ge \frac{1}{256^{\frac{\alpha-\beta}{2}}}.
		\end{equation*}
\end{claim}

\begin{proof}
Assume, for contradiction, that there exists no such $\beta$. That is, $\frac{256^{\gamma+1}H^2_{n_{N^+(v),\gamma}}}{256^{\alpha-1}} < \frac{1}{256^{\frac{\alpha-\gamma}{2}}}$ for all $0 \le \gamma \le \alpha-2$. Then we can sum over all classes:
\begin{align}
	\sum_{\gamma=0}^{\alpha-2}256^{\gamma+1}H^2_{n_{N^+(v),\gamma}} &= \sum_{\gamma=0}^{\alpha-2} \frac{256^{\gamma+1}H^2_{n_{N^+(v),\gamma}}}{256^{\alpha-1}} 256^{\alpha-1}\nonumber \\
	&< 256^{\alpha-1} \sum_{\gamma=0}^{\alpha-2} \frac{1}{256^{\frac{\alpha-\gamma}{2}}}\nonumber \\
	&< 256^{\alpha-1} \label{eqn:heavy}.
\end{align}
However, we also know from maximality of $N^+(v)$ that
\begin{align*}
	2\left( \sum_{\gamma=0}^{\alpha-2}256^{\gamma+1}H^2_{n_{N^+(v),\gamma}} + 256^{\mu+1}H^2_{n_{N^+(v),\mu}+1} - 256^{\mu+1}H^2_{n_{N^+(v),\mu}} \right) \ge \frac{256^\alpha}{56},
	\end{align*}
	which implies that, using the fact that $H^2_{i+1}-H^2_i \le \frac{5}{4}$ for any $i \ge 0$,
\begin{align*}
	2 \sum_{\gamma=0}^{\alpha-2}256^{\gamma+1}H^2_{n_{N^+(v),\gamma}} &\ge \frac{256^\alpha}{56} - 2 \cdot 256^{\mu+1}(H^2_{n_{N^+(v),\mu}+1} - H^2_{n_{N^+(v),\mu}})\\
	&\ge \frac{256^\alpha}{56} - \frac{10}{4} 256^{\alpha-1} > 256^{\alpha-1},
\end{align*}
contradicting Equation \ref{eqn:heavy}.
\end{proof}

Consider all edges of class $\beta$. By using the inequality $H_i \le 1 + \ln i$ and rearranging the equation from Claim \ref{claim:heavy-class}, we get that
\begin{equation*}
	n_{N^+(v), \beta} \ge e^{\sqrt{256^{\frac{\alpha-\beta-4}{2}}}-1}.
\end{equation*}
We charge $c(e_v) \le 256^{\alpha +1}$ equally across all edges of class $\beta$. Therefore each edge of class $\beta$ is charged at most
\begin{equation*}
	256^{\alpha+1} \frac{1}{e^{\sqrt{256^{\frac{\alpha-\beta-4}{2}}}-1}} = 256^\beta \frac{256^{\alpha-\beta+1}}{e^{\sqrt{256^{\frac{\alpha-\beta-4}{2}}}-1}}.
\end{equation*}

Suppose that any other edge $e_{v'}$ of class $\alpha$ is (partially) charged to some edge $a$ that $e_v$ has also been partially charged to. Then $a \in N^+(v)$ and $a \in N^+(v')$ overlap, a contradiction to Lemma \ref{lem:overlap}. Therefore the total amount charged to $a$ is at most
\begin{align*}
	\sum_{\gamma \ge \beta+2} 256^\beta \frac{256^{\alpha-\beta+1}}{e^{\sqrt{256^{\frac{\alpha-\beta-4}{2}}}-1}} &= 256^\beta \sum_{z \ge 0} \frac{256^{z+3}}{e^{\sqrt{256^{\frac{z-2}{2}}}-1}}\\
	&= O(256^\beta)\\
	&= O(c(a)).
\end{align*}
Since each edge $e$ appears in $MC$ at most twice, the total cost of $E^*$ from this type of charging is $O(c(\T))$. This proves the lemma.
\end{proof}

\bibliographystyle{plain}
\bibliography{ref}

\begin{thebibliography}{10}

\bibitem{Albers09}
Susanne Albers.
\newblock On the value of coordination in network design.
\newblock {\em {SIAM} J. Comput.}, 38(6):2273--2302, 2009.

\bibitem{AnshelevichDKTWR08}
Elliot Anshelevich, Anirban Dasgupta, Jon~M. Kleinberg, {\'{E}}va Tardos, Tom
  Wexler, and Tim Roughgarden.
\newblock The price of stability for network design with fair cost allocation.
\newblock {\em {SIAM} J. Comput.}, 38(4):1602--1623, 2008.

\bibitem{BiloB11}
Vittorio Bil{\`{o}} and Roberta Bove.
\newblock Bounds on the price of stability of undirected network design games
  with three players.
\newblock {\em Journal of Interconnection Networks}, 12(1-2):1--17, 2011.

\bibitem{BiloCFM13}
Vittorio Bil{\`{o}}, Ioannis Caragiannis, Angelo Fanelli, and Gianpiero Monaco.
\newblock Improved lower bounds on the price of stability of undirected network
  design games.
\newblock {\em Theory Comput. Syst.}, 52(4):668--686, 2013.

\bibitem{BiloFM13}
Vittorio Bil{\`{o}}, Michele Flammini, and Luca Moscardelli.
\newblock The price of stability for undirected broadcast network design with
  fair cost allocation is constant.
\newblock In {\em FOCS}, pages 638--647, 2013.

\bibitem{ByrkaGRS13}
Jaroslaw Byrka, Fabrizio Grandoni, Thomas Rothvo{\ss}, and Laura Sanit{\`{a}}.
\newblock Steiner tree approximation via iterative randomized rounding.
\newblock {\em J. {ACM}}, 60(1):6, 2013.

\bibitem{ChakrabartyDV11}
Deeparnab Chakrabarty, Nikhil~R. Devanur, and Vijay~V. Vazirani.
\newblock New geometry-inspired relaxations and algorithms for the metric
  steiner tree problem.
\newblock {\em Math. Program.}, 130(1):1--32, 2011.

\bibitem{CharikarKMNS08}
Moses Charikar, Howard~J. Karloff, Claire Mathieu, Joseph Naor, and Michael~E.
  Saks.
\newblock Online multicast with egalitarian cost sharing.
\newblock In {\em {SPAA} 2008: Proceedings of the 20th Annual {ACM} Symposium
  on Parallelism in Algorithms and Architectures, Munich, Germany, June 14-16,
  2008}, pages 70--76, 2008.

\bibitem{ChenR09}
Ho{-}Lin Chen and Tim Roughgarden.
\newblock Network design with weighted players.
\newblock {\em Theory Comput. Syst.}, 45(2):302--324, 2009.

\bibitem{ChristodoulouCLPS09}
George Christodoulou, Christine Chung, Katrina Ligett, Evangelia Pyrga, and Rob
  van Stee.
\newblock On the price of stability for undirected network design.
\newblock In {\em Approximation and Online Algorithms, 7th International
  Workshop, {WAOA} 2009, Copenhagen, Denmark, September 10-11, 2009. Revised
  Papers}, pages 86--97, 2009.

\bibitem{EpsteinFM09}
Amir Epstein, Michal Feldman, and Yishay Mansour.
\newblock Strong equilibrium in cost sharing connection games.
\newblock {\em Games and Economic Behavior}, 67(1):51--68, 2009.

\bibitem{FanelliLMS15}
Angelo Fanelli, Dariusz Leniowski, Gianpiero Monaco, and Piotr Sankowski.
\newblock The ring design game with fair cost allocation.
\newblock {\em Theor. Comput. Sci.}, 562:90--100, 2015.

\bibitem{FiatKLOS06}
Amos Fiat, Haim Kaplan, Meital Levy, Svetlana Olonetsky, and Ronen Shabo.
\newblock On the price of stability for designing undirected networks with fair
  cost allocations.
\newblock In {\em ICALP}, pages 608--618, 2006.

\bibitem{KawaseM13}
Yasushi Kawase and Kazuhisa Makino.
\newblock Nash equilibria with minimum potential in undirected broadcast games.
\newblock {\em Theor. Comput. Sci.}, 482:33--47, 2013.

\bibitem{LeeL13}
Euiwoong Lee and Katrina Ligett.
\newblock Improved bounds on the price of stability in network cost sharing
  games.
\newblock In {\em EC}, pages 607--620, 2013.

\bibitem{Li09}
Jian Li.
\newblock An o(log(n)/log(log(n))) upper bound on the price of stability for
  undirected shapley network design games.
\newblock {\em Inf. Process. Lett.}, 109(15):876--878, 2009.

\bibitem{RajagopalanV99}
Sridhar Rajagopalan and Vijay~V. Vazirani.
\newblock On the bidirected cut relaxation for the metric steiner tree problem.
\newblock In {\em SODA.}, pages 742--751, 1999.

\bibitem{RobinsZ05}
Gabriel Robins and Alexander Zelikovsky.
\newblock Tighter bounds for graph steiner tree approximation.
\newblock {\em {SIAM} J. Discrete Math.}, 19(1):122--134, 2005.

\bibitem{Rosenthal73}
Robert~W Rosenthal.
\newblock A class of games possessing pure-strategy nash equilibria.
\newblock {\em International Journal of Game Theory}, 2(1):65--67, 1973.

\end{thebibliography}

\end{document}